\documentclass[runningheads]{llncs}

\usepackage{amsmath,amssymb,mathtools}
\usepackage{graphicx}
\usepackage{xcolor}
\usepackage{url}
\usepackage{microtype}
\usepackage{enumitem}
\usepackage{placeins}
\usepackage{float}
\usepackage{algorithm}
\usepackage[noend]{algpseudocode}
\usepackage{tikz}
\usetikzlibrary{arrows.meta,positioning,calc,fit,backgrounds,decorations.pathreplacing}
\usepackage[T1]{fontenc}
\usepackage{comment}

\usepackage{orcidlink}
\providecommand{\orcidID}[1]{}
\renewcommand{\orcidID}[1]{\orcidlink{#1}}

\newcommand{\lcp}{\mathsf{LCP}}
\newcommand{\lce}{\mathsf{LCE}}
\newcommand{\strdepth}{\mathsf{strdepth}}
\newcommand{\pref}{\preceq}

\newcommand{\tok}{\mathsf{tok}}
\newcommand{\nil}{\bot}
\newcommand{\repHandle}{\mathsf{rep}}

\newcommand{\key}{\mathsf{key}}

\algrenewcommand\algorithmiccomment[1]{\hfill\textnormal{// #1}}

\title{Efficient LCE Queries and Lexicographic Minimizers on Sliding Suffix Trees}
\titlerunning{LCE Queries and Lexicographic Minimizers on Sliding Suffix Trees}

\author{Toshiharu Minematsu\inst{1} \and Shunsuke Inenaga\inst{2}\orcidID{0000-0002-1833-010X}}
\authorrunning{Minematsu and Inenaga}

\institute{
  {Department of Information Science and Technology, Kyushu University} \\
  \email{minematsu.toshiharu.862@s.kyushu-u.ac.jp}
  \and
  {Department of Informatics, Kyushu University} \\
  \email{inenaga.shunsuke.380@m.kyushu-u.ac.jp}
} 

\begin{document}
\maketitle

\begin{abstract}
We study longest-common-extension (LCE) queries and lexicographic minimizer maintenance on the suffix tree of a sliding window.
The main difficulty is that a sliding suffix tree is maintained in an implicit Ukkonen-style form: some suffixes of the current window are not represented by leaves.
We show that the longest implicit suffix induces a periodic representative map that folds every implicit suffix to an explicit suffix leaf in constant time.
Combined with leaf pointers and dynamic LCA, this yields a linear-space data structure with amortized constant-time window shifts and worst-case constant-time LCE queries over a constant-size alphabet.
For minimizers, the LCE structure gives a direct exact solution, but it uses more machinery than fixed-depth comparisons require.
We therefore give an LCE-free construction.
It materializes each active depth-$k$ locus as an auxiliary frontier node storing a frontier token.
Suffix-leaf handles are created immediately when leaves appear; handles of suffixes shorter than $k$ stay inactive and are activated only when their $k$-mer becomes defined.
Tokens are compared by dynamically maintained lexicographic leaf-order labels
of representative handles, while $k$-mer positions whose suffixes are implicit
are folded by the same representative map.
The common monotone-deque layer over these token pointers then maintains exact lexicographic minimizers without hashing, LCE queries, or dynamic LCA.
\end{abstract}

\section{Introduction}

Suffix trees are classical indexes for exact string processing, and Ukkonen's online construction is the standard starting point for maintaining them incrementally~\cite{Ukkonen1995}.
In a streaming setting, however, the object to be indexed is often only the most recent window, leading to the \emph{sliding suffix tree}.
This line of work goes back to the finite-window data-compression index of Fiala and Greene~\cite{FialaGreene1989}, was refined by Larsson~\cite{Larsson1996}, and was later revisited by Senft~\cite{Senft2005}.
Leonard et al. simplified the maintenance of valid edge labels and leaf pointers~\cite{LeonardEtAl2023}.
Recent work on streaming sliding-window string indexing gives a complementary view of indexing the most recent window under streaming updates~\cite{BilleEtAl2023,BilleEtAl2024}.
We use the sliding suffix tree of Brodnik and Jekovec as the underlying dynamic suffix tree, together with the leaf pointers and valid edge labels of Leonard et al.~\cite{BrodnikJekovec2018,LeonardEtAl2023}.
These primitives are sufficient for the LCE data structure developed in Sections~\ref{sec:representatives} and~\ref{sec:lce}.
For the LCE-free minimizer data structure of Section~\ref{sec:minimizer}, we additionally use the BP-linked suffix tree of Sumiyoshi et al.~\cite{SumiyoshiEtAl2024} together with a standard order-maintenance structure to compare representative suffix leaves in constant time.

This paper studies how to use this representation for two comparison tasks on the current window $S$: LCE queries and lexicographic minimizer maintenance.
For positions $i,j$, $\lce(i,j)$ is the longest common prefix length of $S[i..]$ and $S[j..]$.
For an integer $k$, the lexicographic minimizer is the leftmost occurrence of the smallest current $k$-mer.
Minimizers were introduced as a sparse seed-selection mechanism for biological sequence comparison~\cite{RobertsEtAl2004}, and are closely related to winnowing schemes for local fingerprinting~\cite{SchleimerEtAl2003}.
Their density and ordering have been studied extensively~\cite{MarcaisEtAl2017,MarcaisDeBlasioKingsford2018}.
Both tasks are immediate in a static suffix tree with a terminal symbol, but not in a sliding suffix tree without one: some suffixes are \emph{implicit}, ending inside edges rather than at explicit suffix leaves.
Thus LCE queries and $k$-mer comparisons cannot always be reduced directly to operations on suffix leaves.

LCE and substring comparison have also been studied in the more general dynamic-string setting, where strings are maintained under operations such as split and concatenation~\cite{GawrychowskiEtAl2018,LiptakMasilloNavarro2024}.
Those data structures support much richer updates than sliding-window shifts.
Our focus is different: by exploiting the restricted append/delete nature of the sliding-window model and the maintained sliding suffix tree, we obtain constant-time LCE queries and constant-time fixed-length string comparisons on the current window, apart from the cost of updating the suffix tree itself.

Our first observation is that implicit suffixes in a sliding suffix tree
can be represented by suitable explicit suffix leaves, using a periodic
structure around the longest implicit suffix.  We formalize this by a
representative map in Section~\ref{sec:representatives}, and use it as the main technical device
throughout the paper.

For LCE queries, this map lets us reduce queries involving implicit
suffixes to LCA queries on explicit leaves, followed by a simple truncation
by the true suffix lengths.  This yields constant-time LCE queries on the
current window.


For minimizers, one could directly use these LCE queries to compare $k$-mers.
In addition, we also present an LCE-free alternative data structure, which is possible because minimizer maintenance only needs the relative order of distinct length-$k$ prefixes.


\paragraph{Contributions.}
For sliding windows over constant-size alphabets, we obtain the following results.
\begin{itemize}[leftmargin=1.7em]
    \item We give a periodic representative map that folds every implicit suffix to an explicit suffix leaf in constant time.
    \item We obtain a linear-space data structure with amortized $O(1)$ update time and worst-case $O(1)$ LCE query time by combining this map with leaf pointers and dynamic LCA.
    \item We develop an LCE-free exact minimizer data structure based on auxiliary frontier nodes, BP-linked leaf-order handles, frontier tokens, and the same representative map for implicit suffixes; the standard deque layer is shared with the LCE-based baseline.
\end{itemize}
For general ordered alphabets of size $\sigma$,
the time for updating the sliding suffix tree topology becomes amortized $O(\log \sigma)$.
All the other operations and queries are independent of the alphabet size.

\section{Preliminaries}

\subsection{Strings}
Let $d$ be the fixed window size, and let $S=S[1]S[2]\cdots S[d]$ be the current window string over an alphabet $\Sigma$.
For $1\le i\le j\le d$, $S[i..j]$ denotes $S[i]\cdots S[j]$, and $S[i..d]$ is the suffix with initial position $i$.
For two strings $X,Y$, let $\lcp(X,Y)$ denote the length of their longest common prefix.
The longest common extension of positions $i,j$ is
\[
    \lce(i,j)=\lcp(S[i..d],S[j..d]).
\]

We use suffix trees without appending a unique terminal symbol.
A suffix is \emph{leaf-represented} if it corresponds to a leaf of the maintained tree.
Following the usual Ukkonen-style terminology, we call a suffix \emph{implicit} if it is not leaf-represented; its locus may be either an explicit internal node or a point in the middle of an edge.
The \emph{string depth} $\strdepth(v)$ of an explicit node $v$ is the length of the path label from the root to $v$.
For a point on an edge, the string depth is defined analogously.

A string $Y$ is a border of a string $X$ if it is both a proper prefix and a proper suffix of $X$.
We use the standard fact that if $X$ has a border of length $|X|-p$, then $p$ is a period of $X$.


\subsection{Tools}
We use the following known primitives.
First, sliding suffix trees can be maintained in linear space with amortized constant-time updates over constant-size alphabets, following the Ukkonen-style sliding-window constructions of Fiala--Greene, Larsson, Senft, and Brodnik--Jekovec~\cite{Ukkonen1995,FialaGreene1989,Larsson1996,Senft2005,BrodnikJekovec2018}.
Second, leaf pointers and valid edge labels can be maintained in constant time per leaf insertion/deletion~\cite{LeonardEtAl2023}.
A leaf pointer of a node gives an arbitrary descendant suffix leaf currently present below that node.
Third, the dynamic LCA structure of Cole and Hariharan supports LCA queries and the relevant local tree updates in worst-case constant time~\cite{ColeHariharan2005}.

For minimizers, fix $k\ge 1$.
The $k$-mer at position $i$ is $S[i..i+k-1]$, defined for $1\le i\le d-k+1$.
The lexicographic minimizer of a nonempty set of current $k$-mers is the leftmost occurrence of the lexicographically smallest current $k$-mer.
Throughout the paper, positions such as $i$ and $j$ are written as window-relative positions for readability.
In an implementation of the sliding-window updates, occurrences in queues and deques are stored with absolute stream positions, or equivalently with monotonically increasing timestamps; expiration tests are performed on these absolute positions.
This convention avoids renumbering all stored occurrences after each shift.

\section{Periodic Representatives of Implicit Suffixes}
\label{sec:representatives}

Let $B=S[a..d]$ be the longest non-leaf-represented, or implicit, suffix of the current sliding suffix tree.
The locus of $B$ is called the \emph{active point}~\cite{Ukkonen1995} and is maintained during the sliding suffix tree updates~\cite{Larsson1996,LeonardEtAl2023}.
All implicit suffixes are suffixes of $B$.
This is the standard suffix-chain property of the implicit part of Ukkonen-style suffix trees: the non-leaf-represented suffixes form a suffix chain ending at the shortest current suffix.
Let $v_B$ be the lower explicit endpoint of the edge containing the locus of $B$; if the locus itself is an explicit node, let $v_B$ be that node.
Let $L$ be the initial position of an active descendant leaf below $v_B$, obtained by a leaf pointer.
Then $L<a$: since $S[a..d]$ itself is not leaf-represented, any active descendant leaf below its locus corresponds to a suffix that properly extends $S[a..d]$, and hence has initial position before $a$.

\begin{lemma}[Induced period]
\label{lem:period}
The string $S[L..d]$ has period $\pi=a-L$.
\end{lemma}

\begin{proof}
Since the leaf for suffix $S[L..d]$ is below the locus of $B=S[a..d]$, the string $B$ is a prefix of $S[L..d]$.
The same string $B$ is also the suffix of $S[L..d]$ with initial position $a$.
Thus $B$ is a border of $S[L..d]$.
The period induced by this border is
\[
    |S[L..d]|-|B|=(d-L+1)-(d-a+1)=a-L.
    \]
\qed    
\end{proof}

For an implicit suffix $S[i..d]$, define
\[
    \rho(i)=L+((i-a) \bmod \pi), \qquad \pi=a-L.
\]
Then $L\le \rho(i)<a$.
Since $S[a..d]$ is the longest implicit suffix, every suffix with initial position before $a$ is leaf-represented; hence $\rho(i)$ is indeed the position of an explicit leaf.

\begin{lemma}[Representative leaf]
\label{lem:representative}
For every implicit suffix $S[i..d]$, we have
\[
    S[i..d] \pref S[\rho(i)..d],
\]
where $X\pref Y$ means that $X$ is a prefix of $Y$.
\end{lemma}

\begin{proof}
The positions $i$ and $\rho(i)$ are congruent modulo $\pi$ and both lie in the periodic string $S[L..d]$ of Lemma~\ref{lem:period}.
Therefore $S[i+t]=S[\rho(i)+t]$ for every $0\le t\le d-i$, and the claim follows.
\qed
\end{proof}

The choice of the descendant leaf $L$ is not unique in general; any active descendant leaf below the locus of $S[a..d]$ yields a valid representative by the same periodicity argument.

Define the representative map
\[
\lambda(i)=
\begin{cases}
 i, & \text{if } S[i..d] \text{ is leaf-represented},\\
 \rho(i), & \text{if } S[i..d] \text{ is implicit}.
\end{cases}
\]
Then $\lambda(i)$ is always an explicit leaf position and $S[i..d]\pref S[\lambda(i)..d]$.
The value $\lambda(i)$ is obtained in constant time from $a,L$ and the leaf-represented/implicit status of $i$.
This status is standard information in the sliding suffix tree representation.
Equivalently, once the longest implicit suffix has initial position $a$, all non-leaf-represented suffixes are exactly the suffixes of $S[a..d]$; hence a position $i<a$ is leaf-represented, while a position $i\ge a$ is implicit.
See also Fig.~\ref{fig:representative} for illustration.

\begin{figure}[t]
\centering
\includegraphics[width=0.8\linewidth]{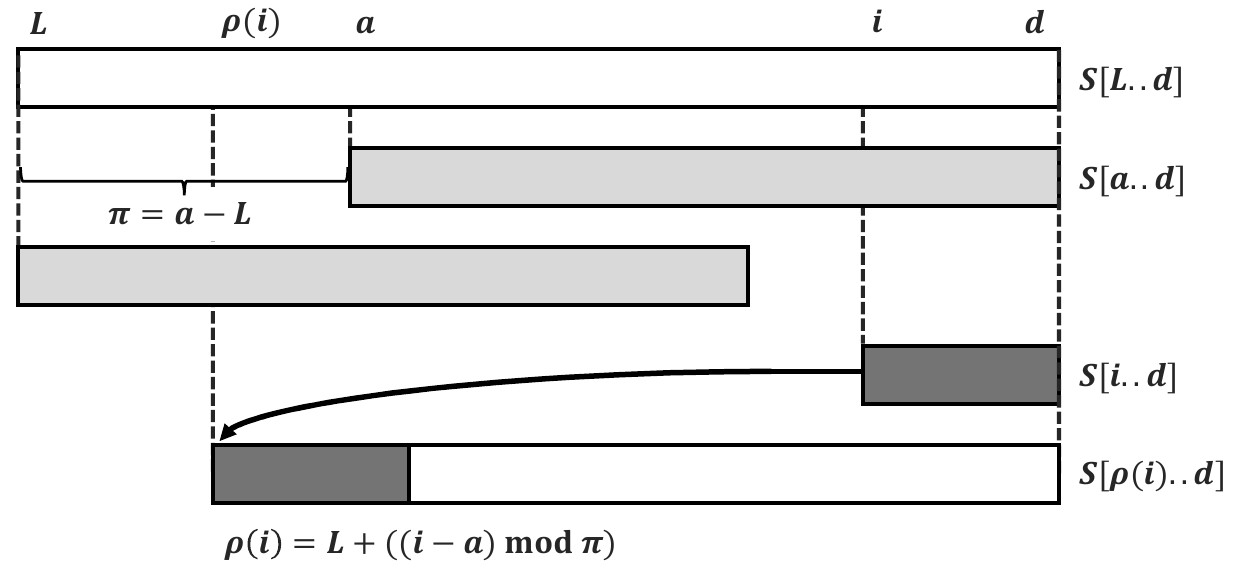}
\caption{The periodic string $S[L..d]$ induced by the longest implicit suffix. The suffix $S[i..d]$ is folded to the explicit representative leaf $S[\rho(i)..d]$ inside the same period class.}
\label{fig:representative}
\end{figure}

\section{LCE Queries and a Direct Minimizer Baseline}
\label{sec:lce}

\subsection{Constant-Time LCE Queries}

To answer $\lce(i,j)$, compute $x=\lambda(i)$ and $y=\lambda(j)$.
Let $u$ be the LCA of the explicit leaves $x$ and $y$, and let $D=\strdepth(u)$.
The answer returned is
\[
    \min\{D,\ d-i+1,\ d-j+1\}.
\]
Algorithm~\ref{alg:lce} gives the corresponding pseudo-code.
The LCA structure and leaf pointers are used as abstract operations, while the periodic representative is computed explicitly from $a$ and $L$.

\begin{algorithm}[t]
\caption{Constant-time LCE query by periodic representatives}
\label{alg:lce}
\begin{algorithmic}[1]
\Require Positions $i,j$ in the current window $S[1..d]$
\Ensure $\lce(i,j)$
\Function{RepresentativeLeaf}{$i$}
    \If{the suffix $S[i..d]$ is leaf-represented}
        \State \Return the leaf of $S[i..d]$
    \Else
        \State $a \gets$ the initial position of the longest implicit suffix
        \State $L \gets$ the initial position of an active descendant leaf below the locus of $S[a..d]$
        \State $\pi \gets a-L$
        \State $r \gets L+((i-a) \bmod \pi)$ \Comment{$r$ is the representative of $i$}
        \State \Return the leaf of $S[r..d]$
    \EndIf
\EndFunction
\Statex
\Function{LCE}{$i,j$}
    \State $x \gets$ \Call{RepresentativeLeaf}{$i$}
    \State $y \gets$ \Call{RepresentativeLeaf}{$j$}
    \State $u \gets \mathrm{LCA}(x,y)$
    \State \Return $\min\{\strdepth(u),\ d-i+1,\ d-j+1\}$
\EndFunction
\end{algorithmic}
\end{algorithm}

\begin{lemma}
\label{lem:lce-correct}
The above value is exactly $\lce(i,j)$.
\end{lemma}

\begin{proof}
By the definition of the representative map, $S[i..d]\pref S[x..d]$ and $S[j..d]\pref S[y..d]$.
The LCA of $x$ and $y$ gives the LCP length $D$ of the representative suffixes.
This common prefix may be longer than a true suffix only because a representative suffix may strictly extend the original implicit suffix.
Thus truncating $D$ by the two true suffix lengths gives exactly the LCP length of $S[i..d]$ and $S[j..d]$.
\end{proof}

\begin{theorem}
\label{thm:lce}
For a sliding window of size $d$ over a constant-size alphabet, there is an $O(d)$-space data structure supporting each shift operation in amortized $O(1)$ time and each LCE query in worst-case $O(1)$ time.
\end{theorem}

\begin{proof}
The sliding suffix tree, the active point, valid edge labels, and leaf pointers are maintained by the known sliding suffix tree machinery~\cite{LeonardEtAl2023}.
The representative leaves are obtained in constant time by Section~\ref{sec:representatives}.
Dynamic LCA gives $u$ in constant time~\cite{ColeHariharan2005}, and the final truncation is constant time.
The maintained structures are all linear in the window size.
\qed
\end{proof}

\subsection{A Generic Deque Layer for Sliding Minimizers}
\label{sec:generic-deque}

We separate the standard sliding-window-minimum part from the way in which
$k$-mers are compared.  Suppose that, for every active valid $k$-mer position
$p$, a data structure maintains a key $\key(p)$ such that the lexicographic
order of two active $k$-mers can be decided by comparing their keys in $O(1)$
time.  Equal keys mean equal $k$-mers.

\begin{lemma}[Generic deque layer]
\label{lem:generic-deque}
Under the above assumption, the leftmost lexicographic minimizer in the current
window can be maintained in $O(d)$ space with amortized $O(1)$ time
per slide and $O(1)$ query time, apart from the cost of maintaining the keys.
\end{lemma}

\begin{proof}
We use the standard monotone deque for sliding-window minima.  Each deque entry
is a pair consisting of a valid $k$-mer position and its current key.  When the
window slides, expired positions are removed from the front.  The newly active
position is then inserted at the back after deleting all suffix entries whose
keys are strictly larger.  Equal keys are not deleted, so the deque preserves
the leftmost tie-breaking rule.  Therefore the front entry is the leftmost
position whose $k$-mer is lexicographically minimum in the current
window.  Each position is inserted once and deleted from the deque at most once,
and every comparison of keys takes $O(1)$ time.
\qed
\end{proof}

\subsection{A Direct LCE-Based Minimizer}
\label{sec:lce-min}

Theorem~\ref{thm:lce} immediately gives exact comparisons of $k$-mers.
For the $k$-mers at positions $i$ and $j$, compute $\delta=\lce(i,j)$.
If $\delta\ge k$, the two $k$-mers are equal.
Otherwise, their order is determined by comparing $S[i+\delta]$ and $S[j+\delta]$.

Thus the LCE structure supplies the constant-time comparison oracle required by
Lemma~\ref{lem:generic-deque}.

\begin{corollary}
\label{cor:lce-min}
Using the LCE structure of Theorem~\ref{thm:lce} together with the generic deque layer of Lemma~\ref{lem:generic-deque}, lexicographic minimizers can be maintained in linear space with amortized $O(1)$ time per shift and $O(1)$ query time over a constant-size alphabet.
\end{corollary}

This solution is simple, but it uses dynamic LCA through the LCE oracle.
The next section builds on this baseline and gives a fixed-depth solution that avoids LCE queries and dynamic LCA.

\section{LCE-Free Lexicographic Minimizers}
\label{sec:minimizer}
\label{sec:kfrontier}

In this section we present an alternative algorithm
for computing lexicographic minimizers without relying on LCE queries.
Namely, we will show the following:
\begin{theorem}
\label{theo:lcefree-min}
There exists an LCE-free algorithm that maintains
lexicographic minimizers in linear space with amortized $O(1)$ time per shift and $O(1)$ query time over a constant-size alphabet.
\end{theorem}

\begin{figure}[tbh]
\centering
\resizebox{0.97\linewidth}{!}{%
\begin{tikzpicture}[x=1.0cm,y=1.0cm,>=Latex,font=\small,
  stnode/.style={circle,draw,fill=white,inner sep=1.2pt},
  leaf/.style={circle,draw,fill=white,inner sep=1.2pt},
  bpbox/.style={rectangle,draw,minimum width=0.46cm,minimum height=0.46cm,inner sep=1pt}]

\node[stnode] (r) at (-2.1,0.0) {};
\node[stnode] (u) at (-3.3,-1.0) {};
\node[leaf]   (l1) at (-4.9,-1.9) {};
\node[leaf]   (l3) at (-1.7,-1.9) {};
\node[leaf]   (l2) at ( 0.3,-1.9) {};
\node[leaf]   (l4) at ( 2.0,-1.9) {};

\draw (r)--node[above left]{$a$}(u);
\draw (u)--node[pos=0.34,left=4pt,fill=white,inner sep=1.4pt]{$bac$}(l1);
\draw (u)--node[pos=0.58,above]{$c$}(l3);
\draw (r)--node[pos=0.42,above left=1pt,fill=white,inner sep=1pt]{$bac$}(l2);
\draw (r)--node[pos=0.58,above]{$c$}(l4);

\node[anchor=east] at (-6.8,-3.25) {BP list:};

\foreach \x/\sym/\val in {
-6.0/(/{\bot},
-5.2/(/{\bot},
-4.4/(/{1},
-3.6/)/{\bot},
-2.8/(/{3},
-2.0/)/{\bot},
-1.2/)/{\bot},
-0.4/(/{2},
 0.4/)/{\bot},
 1.2/(/{4},
 2.0/)/{\bot},
 2.8/)/{\bot}}
{
  \node[bpbox] at (\x,-3.25) {$\sym$};
  \node at (\x,-3.82) {$\val$};
}

\foreach \x/\y in {-6.0/-5.2,-5.2/-4.4,-4.4/-3.6,-3.6/-2.8,-2.8/-2.0,-2.0/-1.2,-1.2/-0.4,-0.4/0.4,0.4/1.2,1.2/2.0,2.0/2.8}
{
  \draw[-{Latex[length=1.4mm,width=1.1mm]},thick] (\x+0.24,-3.12) -- (\y-0.24,-3.12);
  \draw[-{Latex[length=1.4mm,width=1.1mm]},thick] (\y-0.24,-3.38) -- (\x+0.24,-3.38);
}

\draw[dashed] (r.south) .. controls (-5.95,-0.65) and (-6.0,-2.45) .. (-6.0,-3.01);
\draw[dashed] (r.south) .. controls ( 2.75,-0.65) and ( 2.8,-2.45) .. ( 2.8,-3.01);

\draw[dashed] (u.south) .. controls (-5.15,-1.25) and (-5.2,-2.45) .. (-5.2,-3.01);
\draw[dashed] (u.south) .. controls (-1.25,-1.25) and (-1.2,-2.45) .. (-1.2,-3.01);

\draw[dashed] (l1.south) .. controls (-4.95,-2.35) and (-4.4,-2.65) .. (-4.4,-3.01);
\draw[dashed] (l1.south) .. controls (-4.15,-2.35) and (-3.6,-2.65) .. (-3.6,-3.01);

\draw[dashed] (l3.south) .. controls (-1.95,-2.35) and (-2.8,-2.65) .. (-2.8,-3.01);
\draw[dashed] (l3.south) .. controls (-1.75,-2.35) and (-2.0,-2.65) .. (-2.0,-3.01);

\draw[dashed] (l2.south) .. controls ( 0.25,-2.35) and (-0.4,-2.65) .. (-0.4,-3.01);
\draw[dashed] (l2.south) .. controls ( 0.35,-2.35) and ( 0.4,-2.65) .. ( 0.4,-3.01);

\draw[dashed] (l4.south) .. controls (1.75,-2.35) and (1.2,-2.65) .. (1.2,-3.01);
\draw[dashed] (l4.south) .. controls (2.05,-2.35) and (2.0,-2.65) .. (2.0,-3.01);

\end{tikzpicture}%
}
\caption{A BP-linked suffix tree for $S=abac$. Every node contributes an opening parenthesis and a closing parenthesis to the BP list, maintained as a doubly-linked list.
  For the LCE-free minimizer structure, we place a standard order-maintenance data structure on the entire BP list. The handle $h(\ell)$ of a real suffix leaf $\ell$ is the order-maintenance item corresponding to the opening parenthesis of $\ell$. All remaining BP items, including the opening parentheses of internal nodes and the closing parentheses of all nodes, are auxiliary items that carry the comparison key $\bot$.}
\label{fig:bp-linked-om}
\end{figure}
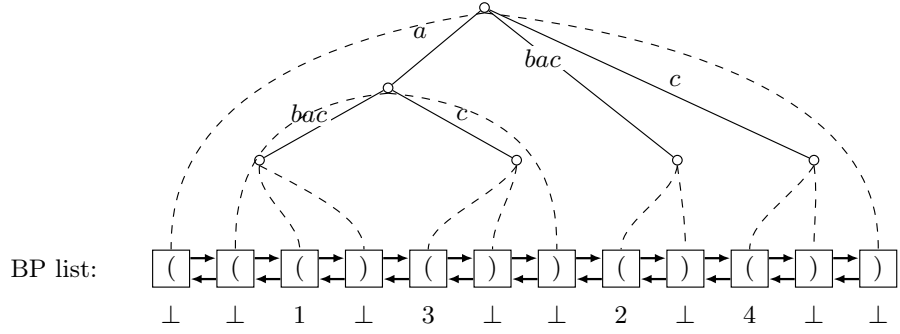
\FloatBarrier

In our LCE-free algorithm,
we use the BP-linked suffix tree of Sumiyoshi et al.~\cite{SumiyoshiEtAl2024} together with a standard order-maintenance structure~\cite{BenderEtAl2002}.
The BP-linked suffix tree maintains a dynamic BP list in which every node contributes an opening parenthesis and a closing parenthesis (see Fig.~\ref{fig:bp-linked-om}).
We mirror all local insertions and deletions in this BP list by the same local updates in a standard order-maintenance structure.
Hence every current BP item receives an order label that can be compared in constant time.
For a real suffix leaf $\ell$, the order item corresponding to the opening parenthesis of $\ell$ is called its \emph{leaf-order handle} and is denoted by $h(\ell)$.
All other BP items, including the opening parentheses of internal nodes and the closing parentheses of all nodes, are treated as auxiliary $\bot$-items.
They stay in the same order-maintenance structure only to support local updates; they are never returned by leaf pointers and never participate in $k$-mer comparisons.

\subsection{Frontier Nodes and Token Lifecycle}

The minimizer problem only compares substrings of length $k$.
We avoid a separate dynamic LCE structure by attaching the fixed-depth information to the order-maintenance labels of the BP items corresponding to real suffix leaves and to the materialized frontier nodes.
The resulting structure has five small components: the BP list with order-maintenance labels, leaf-order handles, frontier tokens, representative folding for implicit suffixes, and representative update after leaf events.
For this section, an object is called \emph{active} if it is currently present in the maintained window structure.
A \emph{frontier class} is the set of current $k$-mer occurrences whose suffixes have the same length-$k$ prefix, and a \emph{frontier token} represents one nonempty frontier class.

\subsubsection{Overview.}
Below we present an overview of our algorithm.

\paragraph{The fixed-depth view.}
For the current window $S[1..d]$, let
$P_k(S)=\{p\mid 1\le p\le d-k+1\}$
be the set of valid $k$-mer positions.
On $P_k(S)$, define the equivalence relation
\[
    p\equiv_k q
    \quad\Longleftrightarrow\quad
    S[p..p+k-1]=S[q..q+k-1].
\]
A frontier token represents one current equivalence class of $\equiv_k$.
For $p\in P_k(S)$, we write $\tok(p)$ for the token of the class containing $p$; for an invalid position $p \notin P_k(S)$ we set $\tok(p)=\nil$.
The goal of the data structure in this section is to maintain the map $p\mapsto\tok(p)$ and the lexicographic order of distinct tokens.  The common deque layer of Lemma~\ref{lem:generic-deque} is then applied to these token keys.

\paragraph{Leaf-order handles.}
Whenever the BP-linked suffix tree creates or deletes BP items, we perform the same local insertions or deletions in the order-maintenance structure.
In particular, when the sliding suffix tree creates a real suffix leaf $\ell$, the opening parenthesis of $\ell$ in the BP list yields the leaf-order handle $h(\ell)$.
This is done regardless of the current suffix length of $\ell$.
The handle is only an order object until the suffix becomes long enough to define a $k$-mer.
A handle is \emph{inactive} while the corresponding suffix has length less than $k$.
During this phase a query for the length-$k$ prefix of that suffix returns $\nil$, and the handle is not exposed to the generic deque layer.
If a sliding update retargets or shortens a leaf to the active point, the same real-leaf BP item remains at the same relative position among the real suffix leaves.
We detach the handle from its old token, update that token if necessary, update the suffix metadata of the handle, and then process the same handle according to the new suffix length.
No new leaf-order insertion is needed in such a retargeting event.

\paragraph{Frontier nodes and tokens.}
Every active length-$k$ locus is materialized as an explicit frontier node and stores one frontier record.
These nodes can be auxiliary (i.e. possibly non-branching) explicit nodes of the minimizer layer: materializing or removing them never changes the BP list or the order-maintenance labels.
Once an inactive handle first becomes long enough to define a $k$-mer, the handle is processed at the corresponding frontier node.
By Lemma~\ref{lem:short-activation} shown below, if the handle was created while the suffix length was less than $k$ and now reaches length exactly $k$, then its length-$k$ prefix forms a new frontier class; hence no existing token for that class is needed in this activation case.
More generally, when a leaf is created or retargeted with suffix length at least $k$, its handle is attached to the token stored at the corresponding frontier node if such a token already exists; otherwise a new token is created and stored there.
When a handle is attached to an existing token whose representative is temporarily undefined, that handle is immediately chosen as the new representative.
Thus no search for the future position of a short suffix is needed: its order position is fixed when the corresponding real-leaf BP item is created.

\paragraph{Accessing frontier nodes.}
It remains to explain how the frontier node needed by a newly reported leaf handle is obtained without a weighted-level-ancestor query.
When a suffix leaf $\ell$ is inserted, the BP-linked-tree update also inserts the corresponding BP items at locally determined positions in the BP list.  We mirror these insertions in the order-maintenance structure and associate the new suffix leaf $\ell$ with an order-maintenance item, called its leaf-order handle $h(\ell)$.  Concretely, this handle is implemented by the order-maintenance item corresponding to the opening parenthesis of $\ell$ in the BP list.
There are two cases.  If the local anchor $z$ has string depth at least $k$, then $z$ already lies below the depth-$k$ locus of an existing frontier class.  Hence the leaf pointer of $z$ returns an existing active suffix leaf in the same frontier class, and we obtain the corresponding frontier node and token from the handle of that leaf.
Otherwise, the new leaf is not inserted into an already existing length-$k$ frontier subtree.  If the suffix represented by the new handle already has length at least $k$, we materialize the depth-$k$ locus on the newly created leaf path and create the corresponding frontier token if necessary.  If the suffix length is still less than $k$, the handle remains inactive until the suffix first reaches length $k$.

Each active frontier token $F$ stores an occurrence queue containing all current occurrence positions of this $k$-mer class and a set $H(F)$ of leaf-order handles currently attached to the same frontier node.
One handle in $H(F)$ is designated as the representative handle $\repHandle(F)$.
For two active tokens $F$ and $G$, we compare them by comparing the order-maintenance labels of $\repHandle(F)$ and $\repHandle(G)$.
This gives the lexicographic order of their $k$-mers because the real suffix leaves with the same length-$k$ prefix form one contiguous interval in left-to-right leaf order, and the relative order of the opening parentheses of the real suffix leaves in the BP list is exactly this leaf order.

\paragraph{Positions with implicit suffixes.}
For a valid $k$-mer position $i$, $\tok(i)$ is obtained from the frontier node of the corresponding length-$k$ prefix.
If $S[i..d]$ is leaf-represented, this token is obtained from the handle of that leaf.
If $S[i..d]$ is implicit, we use the representative leaf $\lambda(i)$ from Section~\ref{sec:representatives}; then
$S[i..i+k-1]=S[\lambda(i)..\lambda(i)+k-1]$
by Lemma~\ref{lem:representative}, and therefore
$\tok(i)=\tok(\lambda(i))$.
If $i\notin P_k(S)$, the suffix is too short to define a $k$-mer and the query returns $\nil$.

\paragraph{Representative update.}
If $\repHandle(F)$ is deleted but $H(F)$ remains nonempty, we replace it by any element in $H(F)$.
If $H(F)$ becomes empty while the occurrence queue of $F$ is still nonempty, we put $F$ on a constant-size update list and restore a representative after expired occurrences have been removed.
A surviving occurrence of $F$ implies that the frontier node of $F$ still has an active descendant leaf; the leaf pointer stored at the frontier node returns such a leaf, whose handle becomes the new representative.
If the occurrence queue becomes empty, the token is deleted and its frontier record is cleared.
Thus, before any token comparison is performed, every remaining active token has a representative handle.

\begin{figure}[t]
\centering
\begin{tikzpicture}[x=0.9cm,y=0.9cm,>=Latex,font=\small,
  frontier/.style={circle,draw,fill=blue!15,inner sep=1.4pt},
  selleaf/.style={circle,draw,fill=green!20,inner sep=1.4pt},
  plainleaf/.style={circle,draw,fill=white,inner sep=1.1pt}]
  \node (r) at (0,0) [circle,fill=black,inner sep=1.2pt,label=left:{root}] {};

  \node (f1) at (-3.2,-2.1) [frontier,label=left:{depth $k$}] {};
  \node (f2) at (0,-2.1) [frontier] {};
  \node (f3) at (3.2,-2.1) [frontier] {};
  \draw (r)-- (f1);
  \draw (r)-- (f2);
  \draw (r)-- (f3);

  \draw[black] (f1) -- (-4.3,-4.0) -- (-2.1,-4.0) -- cycle;
  \draw[black] (f2) -- (-1.1,-4.0) -- (1.1,-4.0) -- cycle;
  \draw[black] (f3) -- (2.1,-4.0) -- (4.3,-4.0) -- cycle;

  \node (l11) at (-4.0,-4.0) [selleaf] {};
  \node (l12) at (-3.2,-4.0) [plainleaf] {};
  \node (l13) at (-2.4,-4.0) [plainleaf] {};
  \node (l21) at (-0.8,-4.0) [plainleaf] {};
  \node (l22) at (0,-4.0) [selleaf] {};
  \node (l23) at (0.8,-4.0) [plainleaf] {};
  \node (l31) at (2.4,-4.0) [plainleaf] {};
  \node (l32) at (3.2,-4.0) [plainleaf] {};
  \node (l33) at (4.0,-4.0) [selleaf] {};

  \draw[dashed,->,black] (f1) -- (l11);
  \draw[dashed,->,black] (f2) -- (l22);
  \draw[dashed,->,black] (f3) -- (l33);

  \node[align=left] at (6.0,-0.4) {frontier tokens};
  \node (t1) at (4.8,-1.8) [frontier,label=below:{$F_1$}] {};
  \node (t2) at (6.0,-1.8) [frontier,label=below:{$F_2$}] {};
  \node (t3) at (7.2,-1.8) [frontier,label=below:{$F_3$}] {};
  \draw[thick] (4.2,-1.8)--(7.8,-1.8);
  \foreach \x in {4.8,6.0,7.2} {\draw[gray!60] (\x,-2.1)--(\x,-1.5);}  

  \draw[->,thick,blue!70] (l11.east) .. controls (-1.0,-3.9) and (3.7,-2.9) .. (t1.south);
  \draw[->,thick,blue!70] (l22.east) .. controls (1.6,-3.9) and (5.2,-2.9) .. (t2.south);
  \draw[->,thick,blue!70] (l33.east) .. controls (4.6,-3.9) and (6.8,-2.9) .. (t3.south);
\end{tikzpicture}
\caption{Depth-$k$ frontier nodes represented by frontier tokens.  Below each frontier node, the triangle denotes the subtree of suffix leaves sharing the corresponding length-$k$ prefix.  A leaf pointer selects one descendant leaf, whose handle serves as the representative of the frontier token.}
\label{fig:kfrontier}
\end{figure}
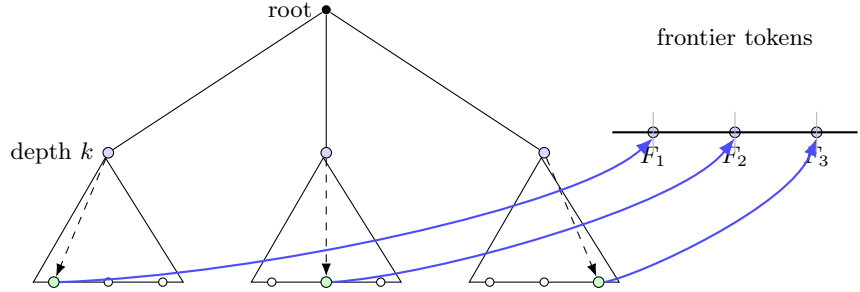

Fig.~\ref{fig:kfrontier} illustrates our data structure.
A pseudo-code is given in Appendix~\ref{app:minimizer-pseudocode}.

\subsubsection{Formal analysis.}
Here we give formal analysis of our algorithm.

\begin{lemma}[Activation of a short leaf]
\label{lem:short-activation}
Let $\ell$ be an active suffix leaf that was created while its suffix length was less than $k$.
If $\ell$ later reaches length exactly $k$, then its length-$k$ prefix forms a new frontier class.
\end{lemma}

\begin{proof}
Suppose that another active suffix has the same length-$k$ prefix.
At the activation moment the suffix represented by $\ell$ has length exactly $k$, so the whole suffix of $\ell$ is a prefix of that other suffix.
Then the locus of $\ell$ would not be an active suffix leaf in the no-end-marker suffix tree, a contradiction.
\qed
\end{proof}

\begin{lemma}[Equality and order by tokens]
\label{lem:token-order}
For valid $k$-mer positions $i,j$,
$\tok(i)=\tok(j)$ iff $i\equiv_k j$.
Moreover, if $\tok(i)\ne \tok(j)$, then the lexicographic order of the two $k$-mers is the order of the representative leaf handles of $\tok(i)$ and $\tok(j)$ in the order-maintenance order.
\end{lemma}

\begin{proof}
For leaf-represented suffixes, equality of the first $k$ characters is exactly equality of their depth-$k$ frontier node, hence of their active token.
For implicit suffixes, Lemma~\ref{lem:representative} gives a representative leaf with the same first $k$ characters, so the same argument applies after folding to the representative leaf.
The leaves with a fixed length-$k$ prefix form a contiguous interval in lexicographic leaf order, and distinct intervals occur in lexicographic order of their prefixes.
Since the order-maintenance structure preserves the relative order of the real suffix-leaf items in the BP list, any representative leaf handle of a token gives the correct order key for that token.
\qed
\end{proof}

\begin{lemma}[Token lifecycle]
\label{lem:token-lifecycle}
The active frontier tokens and their representative handles can be maintained in $O(d)$ space with amortized $O(1)$ update time over a constant-size alphabet.
For two active frontier tokens, their order can be tested in $O(1)$ time.
\end{lemma}

\begin{proof}
All current BP items are maintained in the order-maintenance structure, and the leaf-order handles are exactly the items corresponding to the opening parentheses of real suffix leaves.
A handle is inactive while its suffix length is less than $k$.
When such a handle first reaches length $k$, Lemma~\ref{lem:short-activation} shows that it creates a new frontier token, and no insertion search is necessary because the handle is already present in the maintained order.
If a leaf is created or retargeted with suffix length at least $k$, the token for its length-$k$ prefix is obtained by the local-anchor rule above.
If the supplied local anchor has string depth at least $k$, then its leaf pointer returns an existing leaf in the same frontier class, so the existing frontier node and token are reached from that leaf handle.
Otherwise, if the suffix represented by the handle already has length at least $k$, the depth-$k$ locus is materialized on the newly created leaf path and the corresponding token is created if necessary.
If the suffix length is still less than $k$, the handle remains inactive until its suffix first reaches length $k$.
For an implicit suffix, the representative map of Section~\ref{sec:representatives} gives an explicit leaf in the same class.
When a leaf handle is deleted or detached by a retargeting event, it is removed from the handle set of its old token; if it was the representative and another handle remains, that handle becomes the new representative.
If no handle remains, the token is placed on an update list.
After the expired occurrence, if any, has been removed, a token on the update list with an empty occurrence queue is deleted.
Otherwise the token has a surviving occurrence position $i$.
If $S[i..d]$ is leaf-represented, its leaf is an active descendant of the token's frontier node; if it is implicit, the representative map of Section~\ref{sec:representatives} gives an explicit representative leaf with the same length-$k$ prefix.
Hence the frontier node of a surviving token has an active descendant leaf, and its leaf pointer yields a replacement handle.
This restores $\repHandle(F)$ before the token is compared with any other token.

Each update changes only amortized constantly many real-leaf handles, auxiliary $\bot$-items, and tokens over a constant alphabet, and the same is true for the tokens put on the update list, by the BP-linked suffix-tree updates~\cite{SumiyoshiEtAl2024}.
The order comparison of two active tokens is the order-maintenance label comparison of their representative leaf handles, which is constant time after update.
The maintained handles, active tokens, and occurrence queues are all bounded by the number of active suffixes or active $k$-mer occurrences, hence use $O(d)$ space.
\qed
\end{proof}

\subsection{Wrapping Up}


\begin{proof}[Proof of Theorem~\ref{theo:lcefree-min}]
Lemma~\ref{lem:token-lifecycle} maintains the active frontier tokens and their
representative handles in $O(d)$ space with amortized $O(1)$ update time, and
Lemma~\ref{lem:token-order} gives constant-time comparison of the corresponding
$k$-mers.  Applying Lemma~\ref{lem:generic-deque} to these token keys maintains
the leftmost lexicographic minimizer with an additional $O(d)$ space and
amortized $O(1)$ time per slide.  The construction uses no hashing, no LCE
queries, and no dynamic LCA; comparisons are performed only by
order-maintenance labels of representative handles of frontier tokens.  Hence
the total space is $O(d)$, the amortized update time is $O(1)$ per shift, and
the current minimizer is reported in $O(1)$ time.
\qed
\end{proof}



\newpage

\bibliographystyle{splncs04}
\bibliography{ref}

\begin{thebibliography}{10}
\providecommand{\url}[1]{\texttt{#1}}
\providecommand{\urlprefix}{URL }
\providecommand{\doi}[1]{https://doi.org/#1}

\bibitem{BenderEtAl2002}
Bender, M.A., Cole, R., Demaine, E.D., Farach-Colton, M., Zito, J.: Two
  simplified algorithms for maintaining order in a list. In: ESA 2002. Lecture
  Notes in Computer Science, vol.~2461, pp. 152--164 (2002)

\bibitem{BilleEtAl2023}
Bille, P., Fischer, J., G{\o}rtz, I.L., Pedersen, M.R., Stordalen, T.J.:
  Sliding window string indexing in streams. In: CPM 2023. vol.~259, pp.
  4:1--4:18 (2023)

\bibitem{BilleEtAl2024}
Bille, P., Gawrychowski, P., G{\o}rtz, I.L., Tarnow, S.R.: Faster sliding
  window string indexing in streams. In: CPM 2024. vol.~296, pp. 8:1--8:14
  (2024)

\bibitem{BrodnikJekovec2018}
Brodnik, A., Jekovec, M.: Sliding suffix tree. Algorithms  \textbf{11}(8), ~118
  (2018)

\bibitem{ColeHariharan2005}
Cole, R., Hariharan, R.: Dynamic {LCA} queries on trees. SIAM Journal on
  Computing  \textbf{34}(4),  894--923 (2005)

\bibitem{FialaGreene1989}
Fiala, E.R., Greene, D.H.: Data compression with finite windows. Communications
  of the ACM  \textbf{32}(4),  490--505 (1989)

\bibitem{GawrychowskiEtAl2018}
Gawrychowski, P., Karczmarz, A., Kociumaka, T., {\L}{\k{a}}cki, J., Sankowski,
  P.: Optimal dynamic strings. In: SODA 2018. pp. 1509--1528 (2018)

\bibitem{Larsson1996}
Larsson, N.J.: Extended application of suffix trees to data compression. In:
  DCC 1996. pp. 190--199 (1996)

\bibitem{LeonardEtAl2023}
Leonard, L., Inenaga, S., Bannai, H., Mieno, T.: Constant-time edge label and
  leaf pointer maintenance on sliding suffix trees. arXiv:2307.01412 (2023)

\bibitem{LiptakMasilloNavarro2024}
Lipt{\'a}k, Z., Masillo, F., Navarro, G.: A textbook solution for dynamic
  strings. In: ESA 2024. vol.~308, pp. 86:1--86:16 (2024)

\bibitem{MarcaisDeBlasioKingsford2018}
Mar{\c{c}}ais, G., DeBlasio, D., Kingsford, C.: Asymptotically optimal
  minimizer schemes. Bioinformatics  \textbf{34}(13),  i13--i22 (2018)

\bibitem{MarcaisEtAl2017}
Mar{\c{c}}ais, G., Pellow, D., Bork, D., Orenstein, Y., Shamir, R., Kingsford,
  C.: Improving the performance of minimizers and winnowing schemes.
  Bioinformatics  \textbf{33}(14),  i110--i117 (2017)

\bibitem{RobertsEtAl2004}
Roberts, M., Hayes, W., Hunt, B.R., Mount, S.M., Yorke, J.A.: Reducing storage
  requirements for biological sequence comparison. Bioinformatics
  \textbf{20}(18),  3363--3369 (2004)

\bibitem{SchleimerEtAl2003}
Schleimer, S., Wilkerson, D.S., Aiken, A.: Winnowing: Local algorithms for
  document fingerprinting. In: ACM SIGMOD International Conference on
  Management of Data. pp. 76--85 (2003)

\bibitem{Senft2005}
Senft, M.: Suffix tree for a sliding window: An overview. In: WDS 2005. pp.
  41--46 (2005)

\bibitem{SumiyoshiEtAl2024}
Sumiyoshi, W., Mieno, T., Inenaga, S.: Faster and simpler online/sliding
  rightmost {Lempel--Ziv} factorizations. In: SPIRE 2024. Lecture Notes in
  Computer Science, vol. 14899, pp. 321--335 (2024)

\bibitem{Ukkonen1995}
Ukkonen, E.: On-line construction of suffix trees. Algorithmica
  \textbf{14}(3),  249--260 (1995)

\end{thebibliography}

\newpage

\appendix
\section{Pseudocode for the LCE-Free Minimizer}
\label{app:minimizer-pseudocode}

This appendix gives the detailed pseudocode for the LCE-free minimizer structure of Section~\ref{sec:minimizer}.
The routines are separated into the top-level update procedure, the BP-linked update interface, leaf-event handlers, representative repair, and token/deque operations.
The procedure \Call{UpdateBPLinkedTree}{$c,R$} in Algorithm~\ref{alg:minimizer} is the callback interface to one BP-linked sliding suffix-tree shift: the underlying BP-linked structure performs the standard update, while each reported leaf event invokes the appropriate handler of Algorithm~\ref{alg:token-leaf-events}.
For an insertion or retargeting event, the local anchor $z$ is the child created by an edge split, or the parent to which the leaf is attached; its leaf pointer is used by \Call{LocateFrontierNode}{} to access the corresponding frontier node when the local anchor has string depth at least $k$.

\begin{algorithm}[H]
\caption{LCE-free update of the lexicographic $k$-mer minimizer}
\label{alg:minimizer}
\begin{algorithmic}[1]
\Require New character $c$ and current deque $D$
\Ensure The leftmost lexicographic minimizer in the updated window, if one exists
\Function{UpdateMinimizer}{$c$}
    \State let $q$ be the expired $k$-mer position in this shift, if it exists
    \State initialize the repair list $R \gets \emptyset$
    \State \Call{UpdateBPLinkedTree}{$c,R$}
    \While{$D$ is nonempty and the position at the front of $D$ is expired}
        \State delete the front element of $D$
    \EndWhile
    \If{$q$ exists}
        \State \Call{DeleteOccurrence}{$q$}
    \EndIf
    \State \Call{RepairRepresentatives}{$R$}
    \State let $p$ be the $k$-mer position that has just become valid, if it exists
    \If{$p$ exists}
        \State $\tau \gets$ \Call{GetToken}{$p$}
        \State \Call{InsertOccurrence}{$p,\tau$}
    \EndIf
    \If{$D$ is empty}
        \State \Return undefined
    \Else
        \State \Return the front element of $D$
    \EndIf
\EndFunction
\Statex
\Function{UpdateBPLinkedTree}{$c,R$}
    \State begin one BP-linked sliding suffix-tree shift for appending $c$
    \For{each leaf event $e$ reported during the shift}
        \If{$e$ deletes a suffix leaf $\ell$}
            \State \Call{OnDeleteLeaf}{$\ell,R$}
        \ElsIf{$e$ retargets a suffix leaf $\ell$ with local anchor $z$}
            \State \Call{OnRetargetLeaf}{$\ell,z,R$}
        \ElsIf{$e$ inserts a new suffix leaf $\ell$ with local anchor $z$}
            \State \Call{OnNewLeaf}{$\ell,z$}
        \EndIf
    \EndFor
    \State finish the BP-linked update of leaf-order labels and leaf pointers
\EndFunction
\end{algorithmic}
\end{algorithm}

\begin{algorithm}[H]
\caption{Leaf-event operations for the token lifecycle}
\label{alg:token-leaf-events}
\begin{algorithmic}[1]
\Function{OnNewLeaf}{$\ell,z$}
    \State create a leaf-order handle $h$ for $\ell$
    \State \Return \Call{ProcessHandle}{$h,z$}
\EndFunction
\Statex
\Function{OnDeleteLeaf}{$\ell,R$}
    \State $h \gets$ the leaf-order handle of $\ell$
    \State \Call{DetachHandle}{$h,R$}
    \State delete $h$ from the order-maintenance structure
\EndFunction
\Statex
\Function{OnRetargetLeaf}{$\ell,z,R$}
    \State $h \gets$ the existing leaf-order handle of $\ell$
    \State \Call{DetachHandle}{$h,R$}
    \State update the suffix metadata of $h$ to the new suffix represented by $\ell$
    \State \Return \Call{ProcessHandle}{$h,z$} \Comment{the leaf-order position of $h$ is unchanged}
\EndFunction
\Statex
\Function{DetachHandle}{$h,R$}
    \If{$h$ is attached to an active token $\tau$}
        \State remove $h$ from $H(\tau)$ and clear the token pointer of $h$
        \If{$\repHandle(\tau)=h$ and $H(\tau)$ is nonempty}
            \State $\repHandle(\tau) \gets$ an arbitrary handle in $H(\tau)$
        \ElsIf{$\repHandle(\tau)=h$}
            \State $\repHandle(\tau) \gets \nil$
            \State add $\tau$ to the repair list $R$
        \EndIf
    \EndIf
\EndFunction
\Statex
\Function{LocateFrontierNode}{$h,z$}
    \If{$z\ne\nil$ and $\strdepth(z)\ge k$}
        \State $h_0 \gets$ the handle of a suffix leaf returned by the leaf pointer of $z$
        \State \Return the frontier node of the token attached to $h_0$
    \Else
        \State \Return the newly materialized depth-$k$ frontier node on the path to the leaf of $h$
    \EndIf
\EndFunction
\Statex
\Function{ProcessHandle}{$h,z$}
    \If{the suffix represented by $h$ has length less than $k$}
        \State mark $h$ inactive
        \State \Return $\nil$
    \EndIf
    \State $v \gets$ \Call{LocateFrontierNode}{$h,z$}
    \State $\tau \gets$ the token stored at $v$, if any
    \If{$\tau$ exists}
        \State attach $h$ to $\tau$ and insert $h$ into $H(\tau)$
        \If{$\repHandle(\tau)=\nil$}
            \State $\repHandle(\tau) \gets h$
        \EndIf
    \Else
        \State create a new active frontier token $\tau$ with $H(\tau)=\{h\}$ and $\repHandle(\tau)=h$
        \State store $\tau$ in the frontier record of $v$
    \EndIf
    \State \Return $\tau$
\EndFunction
\end{algorithmic}
\end{algorithm}

\begin{algorithm}[H]
\caption{Representative repair and token deletion}
\label{alg:token-repair}
\begin{algorithmic}[1]
\Function{DeleteToken}{$\tau$}
    \State clear the frontier record of the frontier node of $\tau$
    \State clear the token pointers of all handles in $H(\tau)$
    \State delete $\tau$
\EndFunction
\Statex
\Function{RepairRepresentatives}{$R$}
    \For{each still-existing token $\tau$ in $R$ with $\repHandle(\tau)=\nil$}
        \If{the occurrence queue of $\tau$ is empty}
            \State \Call{DeleteToken}{$\tau$}
        \ElsIf{$H(\tau)$ is nonempty}
            \State $\repHandle(\tau) \gets$ an arbitrary handle in $H(\tau)$
        \Else
            \State $v \gets$ the frontier node of $\tau$
            \State $h \gets$ the handle of an active suffix leaf returned by the leaf pointer of $v$
            \State insert $h$ into $H(\tau)$, attach $h$ to $\tau$, and set $\repHandle(\tau)\gets h$
        \EndIf
    \EndFor
\EndFunction
\end{algorithmic}
\end{algorithm}

\begin{algorithm}[H]
\caption{Token queries and deque operations}
\label{alg:token-deque}
\begin{algorithmic}[1]
\Function{GetToken}{$p$}
    \If{$S[p..d]$ has length less than $k$}
        \State \Return $\nil$
    \EndIf
    \If{the suffix $S[p..d]$ is implicit}
        \State $a \gets$ the initial position of the longest implicit suffix
        \State $L \gets$ the initial position of an active descendant leaf below the locus of $S[a..d]$
        \State $\pi \gets a-L$
        \State $p' \gets L+((p-a) \bmod \pi)$
        \State \Return \Call{GetToken}{$p'$} \Comment{$p'<a$, so the recursive call is leaf-represented}
    \EndIf
    \State $h \gets$ the leaf-order handle of the suffix leaf for $S[p..d]$
    \If{$h$ is attached to an active token $\tau$}
        \State \Return $\tau$
    \Else
        \State \Return \Call{ProcessHandle}{$h,\nil$}
    \EndIf
\EndFunction
\Statex
\Function{InsertOccurrence}{$p,\tau$}
    \State append $p$ to the occurrence queue of $\tau$
    \State store the pointer $\tau$ with the occurrence record of $p$
    \While{$D$ is nonempty and $\tau$ is before $\tok(\mathrm{back}(D))$ by representative labels}
        \State delete the back element of $D$
    \EndWhile
    \State append $(p,\tau)$ to the back of $D$ \Comment{equal tokens are not removed}
\EndFunction
\Statex
\Function{DeleteOccurrence}{$q$}
    \State $\tau \gets$ the token stored with the occurrence at position $q$
    \State delete $q$ from the occurrence queue of $\tau$
    \If{the occurrence queue of $\tau$ is empty}
        \State \Call{DeleteToken}{$\tau$}
    \EndIf
\EndFunction
\end{algorithmic}
\end{algorithm}

\FloatBarrier

\end{document}